\newcommand{\version}{June, 2014}     

\documentclass[12pt]{article}
\usepackage{bbm}
\usepackage{mathrsfs}
\usepackage{a4,amsthm,amsfonts,latexsym,amssymb}
\usepackage{mdwlist}

\swapnumbers 

\theoremstyle{plain}
\newtheorem{thm}{THEOREM}

\newtheorem{define}[thm]{DEFINITION}
\newtheorem{proposition}[thm]{PROPOSITION}

\newcommand{\beq}{\begin{equation}}
\newcommand{\eeq}{\end{equation}}
\def\beqa{\begin{eqnarray}}
\def\eeqa{\end{eqnarray}}

\newcommand{\one}{{\mathbbm 1}}
\newcommand{\Tr}{{\rm Tr}}

\newcommand{\range}{{\rm range}}

\newcommand{\HH}{{\mathscr H}}

\newcommand{\tdt}{\!\cdot\!}

\def\bce{\begin{center}}
\def\ece{\end{center}}
\def\bit{\begin{itemize}}
\def\eit{\end{itemize}}

\date{\small\version}

\begin{document}
\markboth{\scriptsize{ Characterizing Entropy \qquad  \version}}
         {\scriptsize{ Characterization of Entropy \qquad  \version}}

\title{
\bf{Characterizing Entropy in Statistical Physics and in Quantum Information Theory}}
\author{\vspace{8pt} Bernhard Baumgartner$^1$  \\
\vspace{-4pt}\small{Institut f\"ur Theoretische Physik, Universit\"at Wien}\\
\small{Boltzmanngasse 5, A-1090 Vienna, Austria}}

\maketitle

\begin{abstract}
A new axiomatic characterization with a minimum of conditions for
entropy as a function on the set of states in quantum mechanics is presented.
Traditionally unspoken assumptions are unveiled and replaced by proven consequences of the axioms.
First the Boltzmann-Planck formula is derived. Building on this formula, using the Law
of Large Numbers - a basic theorem of probability theory -
the von Neumann formula is deduced.
Axioms used in older theories on the foundations are now derived facts.
\\[10ex]
\\[3ex]

\end{abstract}

\footnotetext[1]{\texttt{Bernhard.Baumgartner@univie.ac.at}}

\maketitle

\newpage


\section{Introduction}\label{intro}

``Entropy'' turns up in three different settings:
Firstly as a thermodynamic quantity, \emph{without reference to statistical mechanics}.
This appearance, which was its historical origin, is still under discussion.
(See \cite{LY00} and references therein.)
In the language of such pure thermodynamic physics
a ``state'' is characterized by macroscopic parameters only,
without reference to microscopic structure.
Secondly, in statistical physics, it is defined
as a function of states of a system, in quantum mechanics of density matrices,
which describe the microscopic, atomistic structure.
Thirdly, in information theory, it is
a measure of information, where there is a priori
no reference to any material structure.
Nevertheless, details of the physical matter used to deal with information
turn up as fundamentally important in practising quantum information theory.
Now it is a basic task, a work still in progress,
to characterize functions of microscopic states dealing with details of structure,
so that they serve well as a basis either for thermodynamics
or for information theory,
although such a microscopic description is in principle an external contribution
to these two settings.

In this paper we postulate axioms, based on simple fundamental reasons,
to characterize such a function.
In a short sequence of proving one proposition and one theorem
- these proofs cover less then one page and are completely elementary -
we get Planck's formula for Boltzmann's entropy of special states.
The second theorem selects von Neumann's entropy formula as the only function
of general microscopic states, which fulfills all the basic scientific desiderata.
Its proof covers little more than one page.

The formula for the entropy of a density matrix,
\beq  \label{neumann}
S(\rho)=S_{vN}(\rho)=-k_B \Tr\rho \ln(\rho),
\eeq
introduced by von Neumann \cite{vN29},
has been successfully used in statistical mechanics
dealing with equilibrium states of large systems,
establishing a quantum-mechanical
statistical basis for thermodynamics.
The enormous success of its use had the effect that questions about basic
justification have been posed only much later, \cite{O75}.
The same is to be said about Shannon's formula, put forward in \cite{S48,SW49} for the entropy of classical information;
see \cite{AFN74} for a ``justification''.

We see several reasons to ``revisit'' the basic axiomatics.
One is to ask for counterparts to Lieb-Yngvason's axioms
for thermodynamic entropy, \cite{LY98,LY99,LY00,LY03}.
Another reason is that the foundations
both of quantum mechanics and of thermodynamics are now seen  in a new light,
different from that of former times.
Therefore, explanations of concepts by referring to our ``natural'' classical intuition,
as in \cite{AFN74,O75},  seem questionable.

The main points of difference of our approach
in relation to the old axiomatic approaches are in detail:
\begin{itemize}
  \item \cite{AFN74} deals with information encoded by using classical probability distributions only.
  Quantum mechanical density matrices are not mentioned.
  \item  Neither \cite{AFN74} nor \cite{O75} mention physical substance as a possible theme.
  This can be seen as the \emph{unspoken assumption} that entropy
  does not depend on the type of matter.
  Here, in the present paper, this property of universality
  is neither postulated nor quietly assumed without mentioning it,
  but it is proven as a consequence of two moderate, but fundamental, axioms.
  A discussion of reasons is presented in the Subsection \ref{equiway}.
  \item Invariance of the entropy function under \emph{general} unitary transformations
  and partial isometries is stated in \cite{O75} as ``Postulate I''.
  I can not see such an assumption as ``natural'' without stating reasons.
  The founding fathers of quantum statistical physics did have reasons,
  which we mention in the penultimate paragraph of Subsection \ref{history2},
  and which we update to modern physics in the Subsection \ref{equiway}.
  These reasons alone do \emph{not yet} give \emph{complete} universality,
  there are still some restrictions, stated in our ``Axiom B''.
  The completeness can be inferred only together with ``Axiom A''.
  \item Additivity and Subadditivity, stated as Postulates II and III
  in \cite{O75}, are well known properties of the von Neumann entropy.
  Hence they emerge here, by way of proving Theorem \ref{vNE}, as consequences
  of other axioms, especially by postulating only extensivity and a weak form of monotonicity.
  \item The classical concept of building up a system from its parts
  is, again unspoken, the principle of the axiomatic in \cite{O75}.
  Here we go another way, using the law of large numbers in probability theory,
  as is explained in the paragraph following this list.
  \item  The possibility of entanglement of the whole system with
  other systems is not discussed in \cite{O75}.
  Now, in our approach, appearance of such an entanglement is allowed in some respect,
  and its effect,
  which disappears in the limit of large numbers, is
  discussed in the present paper, Subsection \ref{flesh}.
  \item It is an important effect of our change of view that we get simpler proofs,
  presented in Subsection \ref{proofs} on only two pages.
  The proofs rely mainly on elementary mathematics,
  nothing but a theorem on relative entropy is needed beyond that.
  The proof in \cite{O75} extends over more than four pages, and builds moreover on the result of \cite{AFN74}.
  The proof in \cite{AFN74} covers more than ten pages and needs results of many other investigations,
  with a lot of citations!
\end{itemize}

Now the present paper is built on the basic concepts of probability
and on its expression in density matrices.
The fundamental relation of probability to the physical world, irrespective of philosophical interpretations,
is through the law of \emph{large numbers}.
So, contrary to considering a state of a system as made up of different parts, we consider
in this paper
unlimited \emph{repetitive appearance} of one and the same state,
drawing a characterization of the single state out of its appearance in large numbers.
That change of point of view also marks the essential difference from other new investigations,
as those starting with \cite{R05}.
Those new investigations treat single systems in a ``non-asymptotic'' regime;
they will be noted in Section \ref{other}.

Both in information theory and in statistical mechanics it makes a difference whether one treats single small systems
(a single short message), whether one looks at very large single systems (long messages),
as Boltzmann did,
or studies large sets of systems (or messages) in the asymptotic regime of large numbers.
In this paper we pass over discussions of single systems
and study entropy from a probabilistic point of view.

After the presentation of formal settings and axioms
a discussion of deeper thoughts on epistemological considerations and reasoning
follows in separate Sections.
Also the backgrounds in history and in more recent developments are included in this discourse,
in Sections \ref{discussion} and \ref{history}.

A word on ``axioms'': A set of axioms can either define a framework for a whole class of theories,
(like Kolmogoroff's axioms for probability),
or it can, in the other extreme, present  postulates building up one special theory
(like Newton's axioms for mechanics or Peano's axioms for the natural numbers).
The Lieb-Yngvason axioms are rather of the first kind.
The axioms presented here, in this paper, express scientific ``desiderata'',
discussed in Section \ref{discussion},
based on fundamental rules of quantum mechanics and probability theory.
These axioms form a sequence of increasingly detailed postulates
leading in several steps to the von Neumann entropy;
but without being explicitly tailored to this goal!

In this paper I have tried to present the logico-mathematical content
precisely and in short,
avoiding as far as possible the use of terms known only to specialists.
The mathematics is separated from more profound reasoning.
The deeper reasons for the axioms are discussed
in Section \ref{discussion}, which is not short;
an indispensable minimum of
arguments, however, has to be stated.
In doing so, I use  arguments
given by the ``founding fathers'' of statistical mechanics
to a great deal, but not exclusively.

\section{Axioms, framework, and consequences}\label{main}

\subsection{Setup}\label{setup}

The framework is that of non-relativistic quantum mechanics.
Each physical system $\mathcal{A}$, $\mathcal{B}$, ... is represented by a separable Hilbert space $\HH_\mathcal{A}$,
 $\HH_\mathcal{B}$,  $\ldots$ ,
states are represented by density matrices $\rho_\mathcal{A}$, $\rho_\mathcal{B}$, $\ldots$ .
Superselection rules like conservation of particle numbers give a decomposition into superselection sectors
$\HH_\mathcal{A}=\bigoplus_\ell \HH_{\mathcal{A},\ell}$,
with restricting density matrices
to be of block-diagonal form, $\rho_\mathcal{A}=\bigoplus_\ell \rho_{\mathcal{A},\ell}$,
forbidding non-zero off-diagonal matrix elements connecting different sectors.

A compound system $\mathcal{A}\cup\mathcal{B}$ consisting of distinguishable subsystems
$\mathcal{A}$, $\mathcal{B}$ is represented by $\HH_\mathcal{A}\otimes\HH_\mathcal{B}$.
A superselection sector for system  $\mathcal{A}$ is defined as a space spanned
by all the eigenvectors belonging to an eigenvalue $q$ of an operator $Q_\mathcal{A}$,
or to a set of eigenvalues $q_\alpha$ of several operators $Q_{\alpha,\mathcal{A}}$,
like particle numbers and charges.
In a compound system the particles, charges, etc. may be distributed over the subsystems.
A sector there is defined as a subspace  characterized by the eigenvalues of
 $Q_\alpha = Q_{\alpha,\mathcal{A}}\otimes\one \oplus \one\otimes Q_{\alpha,\mathcal{B}}$.

Considering compound systems consisting of several versions of the same system,
we denote the n-fold tensor products as $\HH_\mathcal{A}^{\otimes n}$.
If each one of the subsystems is in the same state,
 and when there are no correlations between them,
we denote the composed state as $\rho_\mathcal{A}^{\otimes n}$.

In the course of stating axioms and developing their implications
we use firstly systems with Hilbert spaces where all
eigenvalues of the $Q_{\alpha,\mathcal{A}}$ are specified,
like ``a box with three atoms of gold and one atom of lead''.
We name such a space an ``undivided'' Hilbert space.
Considering transformations of states, either through natural or technically manipulated time evolution,
or as hypothetical reversible  mappings,
each unitary transformation $\rho\mapsto\sigma= U\,\rho\,U^\ast$
on such an undivided Hilbert space is admissible in our axiomatic approach.
Acting in a Hilbert space with several sectors, to be consistent with the superselection rules
the unitaries have to commute with all the $Q_\alpha$.
From standard quantum theory of particles we take the existence
of an undivided Hilbert space with infinite dimension for granted.

We distinguish a special type of density matrix, of  importance in our approach:
\begin{define} \textbf{\textrm{QLB states.}}
A ``Quantum Laplace Boltzmann (QLB) state" is
a state with density matrix $\pi(N):=\frac{1}{N}P_N$, where $P_N$ is
a projector onto an undivided $N$-dimensional space,
a (sub)space of a single superselection sector.
\end{define}

When it comes to the details concerning the law of large numbers, the partial traces are needed
to form the essential one-Hilbert-space density matrix $\omega$ defined on $\HH_k$,
out of a density matrix $\Omega$ defined on an n-fold tensor-product $\HH^{\otimes n}$.
\beq\label{partialstate}
\omega(k):=\Tr_{1,2,\ldots k-1, k+1,\ldots n } \,(\Omega)
\eeq
Finally, to get the best possible kind of continuity (which will turn out to be only a lower semicontinuity),
we refer to the topology of states induced by the trace norm of density matrices.

Now the \textbf{entropy} $S$ shall be a function
on the set of states of systems.
This function $S$ shall obey the following conditions:

\subsection{Stating the Axioms and reasons}\label{reasons}

\subsubsection{Equivalences}\label{equivalences}
We start with the a priori assumption that the entropy of a state might be different
for different substances with equivalent density matrix, and that its functional expression might depend on the
basis in which the state is mathematically written as a density matrix $\rho_\mathcal{A}$,
on the relations of $\rho_\mathcal{A}$ to the Hamiltonian or other observables etc.
Since this hypothetical dependence is only a transitory possibility,
we make no complicated indication of details,
we just mark it by the use of indices $\mathcal{A}, \mathcal{B}$.
\medskip\\
\textbf{AXIOM} \textbf{\emph{A.}}
\emph{If a compound system is considered where one part of it is in a pure state
like $\rho_\mathcal{B}=|\psi_\mathcal{B}\rangle\langle\psi_\mathcal{B}|$,
the entropy of the whole system is determined by the other part (or parts) only,
as if part $\mathcal{B}$ where absent.}
\medskip\\
On the meaning a priori non-universality, on compound systems and decomposition see Section \ref{equiway}.
\medskip\\
\textbf{AXIOM \emph{B.}} \emph{Each unitary transformation of states, which commutes with particle number and other
 operators  related to superselection rules, conserves entropy.}
\medskip\\
On reasons for unitary invariance see remarks in Subsections \ref{equiway} and \ref{history2}.
\medskip\\
The first two axioms (A) and (B) \emph{together} then abolish the assumptions of dependencies on substance and basis.
They imply universality and enable comparisons of all systems
with undivided Hilbert spaces.
We can proof universality: $S$ is a function of $\rho$, irrespective of substance and basis.

\begin{proposition}\textbf{Universality.}\label{universality}
Assuming axioms (A) and (B),
entropy of states for systems with undivided Hilbert space
is a mathematical function
of the spectral values (including multiplicity) of density operators only.
It is also independent of the dimension of the Hilbert space.
\end{proposition}

\subsubsection{QLB states and repeated appearances}\label{QLB}

Two QLB states $\pi(M)$ and  $\pi(N)$ are equivalent in the sense of axioms A and B when $M=N$.
Therefore their entropy is characterized by the dimension of their range only.
It is ``natural'' to consider here strict monotonicity.
Additivity of entropy is practically useful; therefore we pose axiom D.
See remarks in Section \ref{backbone}.
An exclusion of negative infinity as value for the entropy of a pure state is stated for convenience.
Positivity of entropies of all other states is guaranteed without that assumption.
\medskip\\
\textbf{AXIOM} \textbf{\emph{C.}}
\emph{ If $M \, > N$, then the entropy of the
QLB state $\pi(M)$ is strictly larger then the entropy of $\pi(N)$.
The entropy of a pure state is not $-\infty$.}
\medskip\\
\textbf{AXIOM} \textbf{\emph{D.}}
\emph{ A compound system, consisting of $n$ equivalent systems in completely equivalent states,
without any correlation,
shall have $n$ times the entropy of each of its parts.}
\medskip\\
These simple axioms already lead to the Boltzmann entropy:
\begin{thm} \textbf{The Boltzmann-Entropy.}\label{QBP}
Assuming axioms (A \ldots D), and
choosing some constant $k_B$, which is universal,
the Boltzmann-Planck formula for the entropy of QLB states holds:
\beq
S(\pi(N))= k_B\cdot \ln N.
\eeq
\end{thm}

\subsubsection{Generality via Large Numbers}\label{generality}
The properties of entropy shall reflect the law of large numbers.
It implies that repeated appearance of the same state, without correlations,
becomes similar to a certain compound state with discrete uniform distribution.
To specialists this is known as the Asymptotic Equipartition Property, \cite{C06}.
In subsection \ref{flesh} we present details, justifying the following strict condition on the function $S$.
\medskip\\
\textbf{AXIOM} \textbf{\emph{E.}}
\emph{ If the density matrix $\rho$ has finite rank and only rational numbers as eigenvalues,
consider numbers $n$ which are common multiples of their denominators
and consider QLB states $\Omega$ with range in $\HH^{\otimes n}$
which simulate the state  $\rho^{\otimes n}$ in such a way,
that each partial trace $\omega(k)$ of $\Omega$, as defined in equ.(\ref{partialstate}),
gives exactly $\rho$.
The entropy of   $\rho^{\otimes n}$ (which, by axiom (D), is $n$ times the entropy of $\rho$,)
shall dominate the entropy of each of these $\Omega$,
but it shall not be larger then necessary for this dominance.}
\medskip\\
\textbf{AXIOM} \textbf{\emph{F.}}
\emph{ $S(\rho)$ is lower semicontinuous.}
\medskip\\
Those states for which entropy is already defined through axioms (A \ldots  E)
form a dense subset of the set of states for each system.
On Hilbert spaces with finite dimension the remaining infinitesimal gaps
can be filled in in such a way, that $S$ becomes a continuous function.
In case of $\HH$ being of infinite dimension the entropy
$S$ can at best only be extended to a lower semicontinuous function, including $+\infty$ in its range.
This is done by posing axiom F.

These axioms lead to one special entropy function:
\begin{thm} \textbf{von Neumann-Entropy.}\label{vNE}
The only function fulfilling all the axioms is von Neumann's entropy, as written in (\ref{neumann}).
\end{thm}

\subsection{Axioms in mathematical notation}\label{axioms}

\textbf{A}) \textbf{Decomposition}: \quad $S(\rho_\mathcal{A}\otimes|\psi_\mathcal{B}\rangle\langle\psi_\mathcal{B}|)=S(\rho_\mathcal{A})$,
\quad
 $S(|\psi_\mathcal{A}\rangle\langle\psi_\mathcal{A}|\otimes\rho_\mathcal{B})=S(\rho_\mathcal{B})$.
\medskip\\
\textbf{B}) \textbf{Unitary invariance}: \quad  $S(\rho_\mathcal{A})=S(\sigma_\mathcal{A})$
\,\, if \, $\sigma_\mathcal{A}=U\tdt\rho_\mathcal{A}\tdt U^\ast$,
\, $[U,Q_\alpha]=0\,\,\,\forall\,\alpha$.
\medskip\\
\textbf{C}) \textbf{Monotonicity}: \qquad $M \, > N$ $\Rightarrow$ $S(\pi(M)) > S(\pi(N))$,
\qquad
$S(\pi(1)) \neq -\infty$.\\
\medskip\\
\textbf{D}) \textbf{Extensivity, discrete scaling}: \qquad $S(\rho^{\otimes n})=n\cdot S(\rho)$\\
\medskip\\
\textbf{E}) \textbf{Rational combinatorics with large numbers}: \medskip\\
$$
S(\rho):= \sup_{n,\Omega}\left\{\frac 1 n S(\Omega)\,\, |\quad
\range (\Omega)\subset\HH^{\otimes n},\,\, \Omega \,\, \textrm{a QLB state},\quad \forall k\,\,\, \omega(k)=\rho\right\}
$$
\medskip\\
\textbf{F}) \textbf{Semicontinuity}: \\
$$
S(\rho):= \inf_{\sigma(n)} \left\{\lim_{n\to\infty} S(\sigma(n))\,\, |\qquad
\sigma(n)\to\rho\,\, \textrm{in trace norm, as}\,\, n\to\infty \right\}\\
$$

\subsection{Proofs}\label{proofs}

The first two axioms, (A) and (B), imply universality, they enable comparisons of all systems
with undivided Hilbert spaces:
\begin{proof}\textbf{Universality.}
One may choose one system with an infinite dimensional undivided Hilbert space $\HH$ (without an index)
as a \emph{reference} system, and a pure state $|\psi\rangle\langle\psi|$
with $\psi\in\HH$. For any $\rho_\mathcal{A}$ on the undivided $\HH_\mathcal{A}$ one has
$S(\rho_\mathcal{A})=S(\rho_\mathcal{A}\otimes|\psi\rangle\langle\psi|)$,
and this compound state on $\HH_\mathcal{A}\otimes\HH$, which is also undivided,
may be unitarily transformed to $|\psi_\mathcal{A}\rangle\langle\psi_\mathcal{A}|\otimes\rho$
with any $\psi_\mathcal{A}\in\HH_\mathcal{A}$.
One gets $S(\rho_\mathcal{A})=S(\rho)$, now with $\rho$ a density matrix on the referential Hilbert space $\HH$.
This $\rho$ is
equivalent to $\rho_\mathcal{A}$ by an isometry,
and it is characterized only by having the same eigenvalues.
So $S$ does not depend on the type of matter.
\end{proof}%
%
By axioms (A \ldots D) we have excluded any trivial $S$,
and we infer the restriction of $S$ to non-negative values.
Note that $\pi(1)$ is a pure state,
by axiom A it is equivalent to  $\pi(1)\otimes\pi(1)$,
so discrete scaling  gives
$$ S(\pi(1))=S(\pi(1)\otimes\pi(1))=2\cdot S(\pi(1)).$$
Attribution of $\infty$ is impossible, since other $S(\pi(N))$ have to be strictly larger,
$-\infty$ has been excluded a priori,
so $S(\pi(1))=0$ remains as the only allowed value.

Now we can make a ``calibration'' in setting $S(\pi(2))=k_B\cdot \ln 2$,
and prove  Boltzmann's formula (written down by Planck)
for the QLB states, stated in Theorem \ref{QBP},
as a strict consequence of these axioms:
\begin{proof} \textbf{Boltzmann-entropy.}
Observe that $\pi(N)^{\otimes n}=\pi (N^n)$, and
consider numbers $n$, $m=m(n,N)$, obeying $2^m \leq N^n <2^{m+1}$.
Then $m\cdot k_B\cdot \ln2 \leq n\cdot S(\pi(N)) < (m+1)\cdot k_B\cdot \ln2$.
In the limit of large numbers, $N$ fixed, one gets
$$
S(\pi(N))=\left(\lim_{n\to\infty}\frac {m(n,N)} n \right) \cdot k_B\cdot \ln2 = k_B\cdot \ln N.
$$
\end{proof}
\medskip
Proof of Theorem  \ref{vNE}:
\begin{proof} \textbf{von Neumann-entropy.}
\emph{a) $S_{vN}(\rho)$ is an upper bound for the entropy  $S(\rho)$ of each $\rho$ considerd in E:}
The range of $\Omega$ is, by the conditions on the partial traces,
a subset of $(\range(\rho))^{\otimes n}$ which equals $\range(\rho^{\otimes n})$.
For simplicity we restrict the Hilbert space $\HH$ to $\range(\rho)$.
There we define
the operator $L_1$ and on  $\HH^{\otimes n}$ its ``second quantized form'' $L$,
\beq\label{L1}
L_1 :=\ln\rho \,\, ,\qquad
L:=\sum_k L_{1,k},
\eeq
where $L_{1,k}$ acts in the $k^{th}$ factor of $\HH^{\otimes n}$.
The repeated appearance of $\rho$ in every subsystem, $\forall k\in\{1\ldots n\}$: $\omega(k)=\rho$, implies
\beq\label{LOmega}
\Tr(\Omega\cdot L)=n\cdot\Tr(\rho\cdot L_1)=\Tr(\rho^{\otimes n}\cdot L).
\eeq
By Klein's inequality and the \emph{maximum entropy principle}, \cite{W78,W90}, this implies $S_{vN}(\rho^{\otimes n})\geq S_{vN}(\Omega)$.
We present a proof, adapted to the present situation:
Klein's inequality leads to  positivity of relative entropy (section I,B5 in \cite{W78}, or 3.1 in \cite{R02}, or \cite{NC00}):
$$\Tr(\Omega(\ln\Omega-\ln\rho^{\otimes n}))\geq 0.$$
Note that (\ref{L1}) implies $L=\ln\rho^{\otimes n}$,
and therefore, using also (\ref{LOmega}):
$$
S_{vN}(\Omega)=-\Tr(\Omega\ln\Omega)\leq$$
$$\leq-\Tr(\Omega\cdot\ln\rho^{\otimes n})=-\Tr(\Omega L)=-\Tr(\rho^{\otimes n} L)
=-\Tr(\rho^{\otimes n}\ln \rho^{\otimes n})=S_{vN}(\rho^{\otimes n})
$$

Now $S(\Omega)=S_{vN}(\Omega)$ and $S_{vN}(\rho^{\otimes n})=n\cdot S_{vN}(\rho)$.
\medskip\\
\emph{b) For each $\rho$ considered in E the supremum is reached in the limit $n\to\infty$, giving  $S_{vN}(\rho)$:}
Consider the limit of those large numbers $n$ which are multiples of a common denominator of
all the eigenvalues $r_j$ of $\rho$.
Define $\Omega$ as the state proportional to the projector onto
the linear span of those product states formed with eigenstates $\phi_j$ of $L_1$,
where each such $\phi_j$ appears $m_j$ times, where $m_j =n\cdot r_j$.

Each eigenvector $\phi_j$ of $\rho$ can be chosen as an element of a superselection sector,
an eigenvector of all those operators $Q_\alpha$ whose spectral projectors define the sectors.
With $Q_\alpha |\phi_j\rangle=  q_{\alpha,j} |\phi_j\rangle$
each of the chosen product states has the same eigenvalues
$$
Q|\Omega\rangle =  \big(\bigoplus_k Q_{\alpha,k}\big)\bigotimes_k|\phi_{j(k)}\rangle=\sum_j m_j \,q_{\alpha,j}.
$$
These product states are all in the same sector of $\HH^{\otimes n} $,
so the constructed density matrix $\Omega$ is one block, it is a QLB state for which the axioms A \ldots D are valid.

Combinatorics gives $\dim(\range(\Omega))=N$, with the multinomial expression
\beq\label{bsterling}
N= \left( n \atop {m_1 ,  m_2 , \ldots m_\ell }  \right).
\eeq
By Stirling's formula, the limit $n\to\infty$ gives $\lim_{n\to\infty}\frac 1 n \ln N=-\sum_j m_j \ln(m_j)$.
\medskip\\
\emph{c) Filling the infinitesimal gaps:}
The $\rho$ considered in axiom E form a dense subset
in the set of all density matrices.
On this subset the entropy takes on finite values.
But if $\dim(\HH)=\infty$, the values are unbounded,
and not continuous:
For $\rho=\sum_j r_j\cdot |\phi_j\rangle\langle \phi_j |$
consider the sequence
$\rho(N)=\rho-\frac{1}{N}|\phi_1\rangle\langle \phi_1 |\,+\,\frac{1}{N}\pi(N^2)$,
with $N\geq 1/r_1$.
The sequence of von Neumann entropies of $\rho(N)$ diverges.
So von Neumann entropy, which is known to be lower semicontinuous,
can at best be extended as a semicontinuous function,
including $+\infty$ in its range.
\end{proof}

This theorem is a quantum version of the Asymptotic Equipartition Property.
Note that there exists another proof in \cite{T09},
which involves additional statements and is longer.

\section{Discussion}\label{discussion}

\subsection{On the setup}\label{dsetup}

We deal with general states in non-relativistic Quantum Mechanics.
Entropy has in some way to indicate the ``mixedness'' of a mixed state.
As is well known, quantum mechanical mixing is in general different from the classical
mixing of probability distributions, in  that the decomposition into pure states
is not unique.
But there are still cases of classical mixing appearing in quantum mechanics.
Consider for example a single point in a lattice gas of fermions.
There are only two pure states; the lattice point is either empty or occupied.
No coherent superposition is allowed.
Quantumness appears only in a larger set of lattice points,
where there exist pure states of a particle with a wave function
which extends over several points.
But, however large the system may be, there exists no coherent superposition
of states with different numbers of particles.
Such a superselection rule introduces a classicality into Quantum Physics.
So, a priori, we have here to distinguish a fully quantum mechanical system
with all mixed states out of the same superselection sector,
from states where different sectors are involved in the mixing.
It will come only through considerations of large numbers of parts, forming one large system,
that a universal formula for entropy emerges,
making no difference between the kind of mixing.
The a priori possible difference has to be considered when postulating
the invariance of entropy under unitary transformations
done in axiom (B).

A density matrix $\pi(N)$ is the quantized version of distributions introduced by Laplace in the theory of chance
and by Boltzmann in Statistical Mechanics,
see Section \ref{history2}.
Therefore
we call these $\pi(N)$ the ``Quantum Laplace Boltzmann (QLB) states".
We use the QLB states just as special states, discerned from others by their \emph{symmetry}.
This symmetry makes handling them extremely simple.
But that is again a case where we might have to consider superselection rules.
Mixed states with contributions from different sectors do not allow
for as many decompositions into pure states as do mixed states
which live in only one sector. It has less symmetry.
So we start with these special QLB states with the highest possible symmetry.

\subsection{Equivalences, the way from physics to mathematics }\label{equiway}


If we think of an eager freshman in physics with little a priori knowledge,
we guess that this freshman would intend to think that states describing
electrons should carry completely different amounts of entropy compared with
mathematically equivalent states describing nucleons.
The freshman might think:
``If I describe a pure state of an electron with a wave function, and
a pure state of a nucleon with the same wave function, they have different energies.
For each particle has a different constant relating properties of wave functions
to its kinetic energy, namely mass.
Moreover, a mixed state of a particle, being mixed of two pure states with the same
momentum but in different directions,
has another distribution of energy then a mixture of two pure states with different absolute values of momentum.
So, why should such dependencies on details not appear in the formula for entropy?''
But we can safely answer that
such a difference does not appear, and we can state reasons for that fact!
Entropy is a physical property, described by a mathematical function,
independent of the kind of matter, which is described by a density matrix,
and invariant under unitary transformations.
We regard this fact as remarkable and think it is worthwhile to
state it, discuss and derive it from axioms based on solid reasoning.

Decompositions of systems are, at least since Galileo Galilei, at the heart of modern physics.
At this point we pose the minimal postulate
that parts of a compound system, which are in pure states,
have no effect on the entropy, and may as well be absent.
As an example one may think of a compound system, consisting of two boxes,
one with a million atoms of gold, the second one containing one atom of lead in a pure state.
Thinking of applications in thermodynamics, such a part in a pure state can represent,
in some simplification, a weight which can be heaved up or lowered in a gravitational field.

Equivalence is the point where unitary evolutions and transformations of states come into play.
Physics now has the possibility of steering time evolution on an atomistic scale,
\cite{CTGO11}, without observing contradictions to the Second Law of Thermodynamics.
Entropy does not change in the course of invertible time evolutions,
so we postulate that each unitary transformation which is consistent with superselection rules
does not change entropy.

Comparability of states is an essential feature,
it is the central point in the analysis of the second law and entropy in \cite{LY99}.
While Lieb and Yngvason postulate first the existence of an order relation between states
and regard ``adiabatic equivalence'' as derived,
($\rho\sim\sigma$ if both $\rho\prec\sigma$ and $\sigma\prec\rho$),
the procedure in the present paper starts with equivalences, axioms (A) and (B).
We give here no explicit definition of a mathematical equivalence relation;
talking now about equivalence of states, we refer to the applicability of these two axioms
(which could be used for a precise definition).
Comparison, inequalities for the entropy function (axiom C), is then formulated for special equivalence classes.

Axiom (B) paves the way for comparison
by establishing equivalences of states in one sector, as, f.e.
states of three atoms of gold in one box.
Combining now the equivalences of unitarily transformed states (axiom B) with
forming and reducing compound systems by adding
and discarding pure states of another system (axiom A)
lays threads of equivalent states through
all sectors. A system with three atoms of gold has states which are equivalent
to states of a reference system with one atom of lead,
and again equivalent to some states of a system with two atoms of gold in a box.

\subsection{QLB states form the backbone}\label{backbone}

QLB states are completely characterized by one discrete parameter, so they show a truly ``natural''
ordering, allowing for a ``natural'' inequality between their entropies.

Axiom (D), on strict rules for the entropy function in case of repeated uncorrelated appearance
of parts, is the first step to the law of large numbers.
It is here, that negativity is excluded:
$S(\pi(N^2))=2\cdot S(\pi(N))$ and monotonicity imply $S(\pi(N)>0$.
Only the case $S(\pi(1))=-\infty)$ has to be excluded
a priori.

Note that at this point the Boltzmann-Plank formula is derived for $\pi(N)$ states only,
with the assumption that their range is in one sector.
The probabilistic distribution over different sectors has to be analyzed
by using the law of large numbers.

Choosing the extensivity may be considered as convenient, but not absolutely necessary.
For example, demanding $S(\rho^{\otimes n})=(S(\rho))^n$ instead,
would lead to a different, but closely related formula,
which has a strict relation to the Boltzmann-Planck formula.

\subsection{Flesh on the backbone by the law of large numbers}\label{flesh}

We consider states from the point of view of probability theory.
Directly justifying the interpretations
of entropy as a measure of uncertainty or of information is difficult.
(This has been noted already in \cite{BR81}, chapter 6.2.3.)
Nevertheless, entropy arises naturally in the formalism of the law of large numbers.
This law is rarely addressed explicitly in the context of defining entropy,
but it appears implicitly, when ``average information''
is mentioned (f.e. in \cite{NC00}).
It is basic to application of probability theory.
(Jacob Bernoulli, who established the Law of Large Numbers, see f.e. \cite{B01}, regarded
the derivation of this theorem as a greater achievement
than if he had shown how to square a circle, since a proof of the latter would have been of little use.)

This law says, roughly, that in a typical series of $N$ measurements the number $n_i(N)$ of
appearances of event $i$ lies most probably in an interval
\beq\label{largenumberinterval}
[N-c\surd N, N+c\surd N]\cdot\rho_i.
\eeq
with $\rho_i$ the probability for appearance of event $i$, and $c$ some constant.
Our axiom (E) seems, at a first glance at simple examples, to deviate in two ways from this law.
It accounts for a number of possible series which seems to be in one sense too small,
in another sense too large.
The number seems to be too small, since we restrict the  $n_i(N)$
to one definite value, instead of letting it vary in the interval (\ref{largenumberinterval}).
But, taking the logarithm, this difference becomes negligible.
Making a step from $N$ to $2N$, all those series covered by the interval  (\ref{largenumberinterval})
appear in a series of doubled length, where $n_i(2N)=2N\rho_i$ exactly.
In the Boltzmann-Plank formula the logarithm is taken, and $\ln 2$ is negligible for large $N$.

On the other hand, the dimension of QLB-states appears as being too large,
as quantum correlations between the repeated appearances are allowed.
As an example consider a spin up - spin down probability distribution $2/3,1/3$
and a series of $N=3$ tests.
In the $8$-dimensional Hilbert space, a $3$-dimensional subspace is accounted for
by combinatorics, as in formula (\ref{bsterling}).
It is spanned by the vector $\uparrow\uparrow\downarrow$ and its permuted analogues.
The QLB state, according to the conditions of axiom E, allows for one more state vector,
namely $(2/3)^{1/2}\uparrow\uparrow\uparrow + e^{i\alpha}(1/3)^{1/2}\downarrow\downarrow\downarrow$,
where $\alpha$ is an arbitrary phase.
Since already the restriction of the $n_i(N)$ to strict values introduces a kind of
classical correlation between different appearances, I see no reason to exclude
a quantum correlation.
But these extra states can not lead to an excess over the von Neumann entropy,
as is shown in part \textit{a} in the proof of Theorem \ref{vNE}.

\section{Remarks on history and philosophy and other approaches}\label{history}

\subsection{History}\label{history2}

How did the founding fathers proceed, and is there a relation
between their thoughts and the present point of view?
Of the fathers of Statistical Physics, Clausius, Maxwell, Boltzmann (at the time of the precursor
Daniel Bernoulli a concept of entropy was not yet in existence), Boltzmann was the first one who tried to give
an expression for thermodynamic entropy by atomistic statistical terms.
It is not easy to follow his thoughts in depth, see \cite{K73,GRY08}. But
he, and then Gibbs, Planck, Einstein and von Neumann, \cite{B71,B72,B77,G02,P01,P06,E14,vN29},
each one of them gave a justification of his theory by demonstrating its applicability
in thermodynamics, mainly concerning an ideal gas.
The title of von Neumann's paper, \cite{vN29},
``Thermodynamik quantenmechanischer Gesamtheiten'', already shows us von Neumann's intention:
With this formula one establishes thermodynamics dealing with quantum mechanical ensembles,
using Gibb's \cite{G02} methods, transferring them from classical to quantum mechanics.
Even in his book on the mathematical foundations of quantum mechanics, \cite{vN32},
entropy as a function of density matrices is presented only in the subchapter V2
``Thermodynamische Betrachtungen''.

Now the usefulness of their atomistic formulas giving entropy in Statistical Physics
is without any doubt.
But, in the hindsight, one may question the uniqueness:
Is it possible to find another formula, working as well as the existing one
in giving a probabilistic basis for thermodynamics?
That's one of the questions related to
the analysis presented in this paper.
Other questions may be asked in the emerging theory of Quantum Information.
There is no inductive derivation of a quantity like entropy of quantum information from experiments and practice.
But one can state ``desiderata'' how to characterize an appropriate measure
for information carried by quantum states.
I guess, they should be closely related or identical to those posed in Statistical Physics.
Compare \cite{IKO97}.

In the vast universe of probability distributions
and of mixed quantum states there appears a special kind of mixing:
Laplace, in establishing basic methods for dealing with probability, considered sets of events with equal probabilities,
because of a principle of insufficient reason, later named ``the principle of indifference''.
Boltzmann, considering distribution functions on phase space,
performing a careful analysis of compatibility with  Hamilton-Jacobi theory of classical mechanics,
used \emph{equal a priori probability} distributions on sets of given energy.
Now,  a QLB state $\pi(N)$ is the quantized version of such a distribution.

In analogy to the classical theories, one may interpret a QLB state
as ``each pure state $|\psi\rangle\langle\psi|$ characterized by a vector $\psi\in\range(\pi(N))$
has the same probability''.
But, as stated above, we do not need and we do not use
such philosophical preferences.
We use the property of complete symmetry and invariance with respect to any
reversible continuous quantum mechanical time evolution.
We do not a priori assume invariance of entropy under permutations, which is necessary in classical theory,
so we have the condition on its carrier in the Hilbert space to be undivided,
to be located in one superselection sector.

Already prior to Heisenberg's and Schr\"{o}dinger's creation of Quantum Mechanics
it has been stated by Einstein and Szilard, that the time evolution of mixed states
should be considered with the assumption that pure states evolve into pure states.
Einstein, in \cite{E14}, calls this ``Ehrenfests Adiabatenhypothese'';
von Neumann in \cite{vN29} cites this and also page 777 in \cite{S24},
where Szilard refers to his doctoral thesis,
as precursors to his formulation of entropy.
At the time of creating quantum physics only time evolution as is given by nature
had to be considered.
Now we have the possibility of steering time evolution on an atomistic scale,
\cite{CTGO11}, without observing contradictions to the Second Law of Thermodynamics.
So we postulate axiom (B).

The procedure of deriving the $\rho\log\rho$-formula from equal a priori distributions is not new.
There are only differences in the settings.
Boltzmann considers atoms in a gas, Gibbs shows the equivalence of canonical and microcanonical ensembles.
Each of the founding fathers used expression (\ref{bsterling}) and Stirling's formula;
and also did Schr\"{o}dinger in \cite{S46}.

\subsection{Infinities}

Regarding in a short detour extreme philosophical considerations,
whether the world should be modeled with features of infinity, or not,
there appears the case that
we should consider models of a world where only sectors with finite dimension exist.
Now even if the elementary systems have only one-dimensional sectors, as in the example of a lattice gas with fermions,
quantum mechanics comes into play.
This happens through extensions in the compound systems.
One just has to exclude that the world is completely trivial, being only one-dimensional.
We could actually perform our considerations with a model of the world
where there exist only sectors of finite dimension, but without an upper bound.
But, for simplicity, we assume the existence of a sector with infinite dimension.
This sector can serve as the reference Hilbert space.

In classical mechanics entropy can take on all values, between minus and plus infinity
and including these values.
The desiderata, as expressed in the axioms, do actually lead to
strict positivity of the quantum mechanical entropy function,
with the exception of pure states, see subsection \ref{backbone}.

\subsection{Other axioms and other formulas for ``entropy''}\label{other}

The author did not see a single paper on entropy as a function of states where the universality, as stated
in Proposition \ref{universality}, had not been an unspoken assumption.
But, as studied in \cite{LY99}, there is a related nontrivial property of ``comparison'' in thermodynamics!
In statistical mechanics the quiet assumption of universality
trivially implies the validity of axioms (A) and (B).
We unveiled this quiet assumption and found minimal reasons for
validity of universality.

QLB states are completely characterized by one discrete parameter, so they show a truly ``natural''
ordering, allowing for a ``natural'' inequality between their entropies.
This inequality can also be supported by considering certain irreversible evolutions in physics,
namely ``mixing''.
If $M<N$, the QLB state $\pi(N)$ can be constructed by mixing several
unitarily transformed $\pi(M)$.
Thinking of thermodynamics,
irreversibility has to be represented in an increase of entropy.
These ideas are the basis for Uhlmann's order relation for density matrices,
and could be used for a slightly different set of axioms,
see Section \ref{uhl}.

In the last years there has been much progress
in dealing with entropy of single systems, \cite{R05,T09,A10,D11,D12,T12,A13}.
These works introduced some new concepts, like ``operational quantities'' \cite{D12}, ``smooth entropy'' \cite{R05},
``one-shot entropy approach'' \cite{D11,E13}, ``information spectrum approach'' \cite{NH07}.
All these considerations have in common that,
in the asymptotic regime of large sets of mutually independent systems,
they lead to the von Neumann or Shannon entropy.
These limit properties are cases of the asymptotic equipartition property, \cite{C06},
extended to the quantum scenario, \cite{T09}.
In the present paper we do not discuss single systems in one shot.
We deal with appearances of single systems in an unlimited number.
Regarding the Second Law of Thermodynamics, such a repeated appearance should
be discussed when dealing with an attempt to construct a ``perpetuum mobile''.

Also the development of quantum engineering, handling of nano-physics and
emergence of quantum-information leads to a new interest in basic theories
of entropy-like functionals,
characterizing states according to their ``mixedness'',
giving some kind of distance to the pure states.
Such characterizations need not a priori to be in relation to thermodynamics
and they are about small systems.

Relations between various earlier existing types of entropy have been thoroughly
discussed in \cite{W78}. The appearance of new entropy-like expressions
in statistical mechanics and in quantum information theory makes it necessary to check
their usefulness, according to the general desiderata.
In the setup of this paper and as a consequence of the axioms,
other entropy-like formulas are now excluded. They can not serve as a \emph{general} fundamental
element of statistical mechanics or of quantum information.
Nevertheless other ``entropies'', such as the classical mechanical version,
can be useful as limiting or approximating versions.
They will not respect all the desiderata in general, but they may do so under
restricted conditions. So one may apply the present investigation,
by checking which axioms are not valid, which rules are broken in a special case.
To indicate a well known  important example we look at the classical entropy:
In the classical limit it emerges out of von Neumann entropy
as a functional of probability densities, defined with respect
to an a.c. measure.
This functional can take on negative values and is not bounded from below.
Checking there the validity of the axioms,  one may change the setup (Section \ref{setup})
in such a way that analogies of all the axioms of  Section \ref{axioms} are valid, except of axiom (A).
This axiom refers to the appearance of pure quantum states,
it can not be transformed into a classical version:
A possible analogy for pure states are point measures,
but they have negative infinite entropy -
in case one  extends the definition of entropy to such measures.
A remnant of quantum mechanics gives the
natural measure on phase space, with appropriate powers of Planck's constant
as  units of phase space volume.
If one gets a negative value of classical entropy
in a special calculation, it is a sign that in this case
the classical description is not appropriate, energy or temperature is too low,
or a particle density is too high, and so rules given by uncertainty relations are broken.

\subsection{Uhlmann's order relation}\label{uhl}

States allow for a partial ordering related to mixing. It is  Uhlmann's order relation, \cite{U71,U72,U73,W74}.
Note, that this order relation is not the ordering defined in \cite{LY98}.
In matrix analysis \cite{B97} Uhlmann's order relation is known as majorization:
\\
 $\rho\succ\sigma$ means: $\rho$ is more mixed than $\sigma$; $\rho$ is majorized by $\sigma$.
The definition is
\beq
\rho\succ\sigma \qquad \textrm{iff} \qquad \forall N \quad\sum_{n=1}^N r_n \leq \sum_{n=1}^N s_n,
\eeq
where the $r_n$ and the $s_n$ are the eigenvalues of $\rho$ and of $\sigma$ in decreasing order,
$r_1 \geq r_2 \geq r_3 \ldots$,
and at least one of the inequalities is strict.
(Note that $\pi(M)$ is more mixed than $\pi(N)$ if $M>N$.)
With this ordering one could use different axioms:
\medskip\\
\textbf{C'}) \textbf{Monotonicity}: \emph{\quad $\rho\succ\sigma$ $\Rightarrow$ $S(\rho)\geq S(\sigma)$,
\quad strict for $\sigma$ with finite rank.}
\medskip\\
In matrix analysis the ``monotonicity'' is known as ``Schur-concavity''.
Also this way we have  excluded any trivial $S$,
as above we conclude
$S(\pi(1))=0$.
All other states are more mixed and must have larger entropy.
For states not of QLB form, there are still several possibilities, the first four axioms
do f.e. not yet exclude Renyi-entropies.
So it is again the law of large numbers which leads strictly to von Neumann's formula.
The axiom F of semicontinuity can then be replaced by
\medskip\\
\textbf{F'}) \textbf{}: \qquad $S(\rho) := \sup_\sigma \{S(\sigma)\,|\,\rho\succ\sigma\}$.

\subsection*{Thanks}\nonumber

The author thanks the referees for important hints,
and he gives many thanks to Elliott Lieb and Jakob Yngvason for discussions.


\end{document}